\documentclass[reqno]{amsart}

\newcommand{\Rm}{\mathbb{R}}
\newcommand{\Cm}{\mathbb{C}}

\newcommand{\Sm}{\mathbb{S}}
\newcommand{\ba}{\begin{eqnarray*}}
\newcommand{\ea}{\end{eqnarray*}}
\newcommand{\be}{\begin{equation}}
\newcommand{\ee}{\end{equation}}
\newcommand{\bea}{\begin{eqnarray}}
\newcommand{\eea}{\end{eqnarray}}
\newcommand{\va}{\varphi}

\newcommand{\vv}[1]{\boldsymbol{\mathrm{#1}}}
\newcommand{\hvv}[1]{\boldsymbol{\hat{\mathrm{#1}}}}
\newcommand{\uv}{\boldsymbol{{\hat{\Omega}}}}
\newcommand{\uvk}{\boldsymbol{\hat{\mathrm{k}}}}

\newcommand{\pint}{\:\mathcal{P}\!\!\int}
\newcommand{\rrf}[1]{\mathop{\mathcal{R}_{{#1}}}}
\newcommand{\irrf}[1]{\mathop{\mathcal{R}_{{#1}}^{-1}}}

\usepackage{amsmath}
\usepackage{amsthm,amssymb,mathrsfs}

\newtheorem{thm}{Theorem}[section]

\theoremstyle{remark}

\title[]{Three-dimensional transport theory via one-dimensional transport 
theory}

\author[]{Manabu Machida}
\address{Department of Mathematics, University of Michigan}
\email{mmachida@umich.edu}
\thanks{The manuscript was prepared for the special issues of Journal of 
Computational and Theoretical Transport for the 24th International Conference 
on Transport Theory.}

\begin{document}

\begin{abstract}
In linear transport theory, three-dimensional equations reduce to 
one-dimensional equations by means of rotated reference frames. In this paper, 
we illustrate how the technique works and three-dimensional transport theories 
are obtained.
\end{abstract}

\maketitle

\section{Introduction}

In linear transport theory, three-dimensional equations reduce to 
one-dimensional equations with rotated reference frames when the angular 
flux has the structure of separation of variables. Although the technique can 
be used for anisotropic scattering in the presence of boundaries, calculation 
becomes complicated in such general cases and sometimes the essence 
is buried in straightforward but lengthy and tedious calculations. 
In this paper, aiming at elucidating how three-dimensional transport theory 
can be developed using one-dimensional transport theory, we will consider 
the case of isotropic scattering in an infinite medium.

To the best of the author's knowledge, rotated reference frames were first 
introduced in transport theory by Dede \cite{Dede64} in the context of the 
$P_N$ method, and then Kobayashi \cite{Kobayashi77} expanded Dede's 
calculation. Forty years after Dede's finding, Markel \cite{Markel04} devised 
an efficient numerical algorithm of computing solutions to the 
three-dimensional transport equation by reinventing rotated reference 
frames. The method is called the method of rotated reference frames. A lot of 
numerical calculations \cite{Panasyuk06,Xu06a,Xu06b,Machida10,LK11JPA,LK11PRA,
LK11PRE,LK12BOE,LK12JOSA,LK12JPA,LK12JQSRT,LK12PRE,LK13ANE,LK13SR,LK13WRCM,
LK15AO} done during the last decade 
have proved the usefulness and efficiency of this method. It was then found 
that the technique of rotated reference frames is not merely for the 
particular numerical method but is a tool to build bridges between 
three-dimensional transport theory and one-dimensional transport theory. 
Case's method \cite{Case60,Mika61,McCormick-Kuscer66,Case-Zweifel} was 
extended to three dimensions \cite{Machida14}, and the $F_N$ method 
\cite{Siewert78,Siewert-Benoist79} was extended to three dimensions 
\cite{Machida15a}. Recently, the angular flux of the three-dimensional 
transport equation with anisotropic scattering was computed using the Fourier 
transform by making use of rotated reference frames \cite{Machida15b}. 
The technique of rotated reference frames is also applied to 
optical tomography \cite{Schotland-Markel07,MachidaXX,Machida15c}

Let us consider the following transport equation in 
a three-dimensional infinite medium.
\be
\left(\uv\cdot\nabla+1\right)\psi(\vv{r},\uv)
=\frac{\varpi}{4\pi}\int_{\Sm^2}\psi(\vv{r},\uv)\,d\uv+S(\vv{r},\uv),
\label{rte}
\ee
where $\vv{r}\in\Rm^3$, $\uv\in\Sm^2$, and $\varpi\in(0,1)$ is the albedo for 
single scattering. Here, $\psi(\vv{r},\uv)$ is the angular flux and 
$S(\vv{r},\uv)$ is the source term. The unit vector $\uv$ has the polar angle 
$\theta$ and azimuthal angle $\va$. We let $\mu$ denote the cosine of 
$\theta$, i.e., $\mu=\cos\theta$.

Let $\uvk\in\Cm^3$ be a unit vector such that $\uvk\cdot\uvk=1$. For a 
function $f(\uv)$, we introduce an operator $\rrf{\uvk}$ such that 
$\rrf{\uvk}f(\uv)$ is the value of $f$ in which $\uv$ is measured in the 
reference frame whose $z$-axis lies in the direction of $\uvk$. For example, 
when $f(\uv)=\mu$, we have
\[
\rrf{\uvk}\mu=\uvk\cdot\uv.
\]
Such reference frames which are rotated in the direction of a given unit 
vector $\uvk$ are called rotated reference frames. In general, we can compute 
$\rrf{\uvk}f(\uv)$ using Wigner's $D$-matrices \cite{Varshalovich} as
\[
\rrf{\uvk}f(\uv)=\sum_{l=0}^{\infty}\sum_{m=-l}^lf_{lm}\sum_{m'=-l}^l
D_{m'm}^l\left(\va_{\uvk},\theta_{\uvk},0\right)Y_{lm'}(\uv),
\]
where $\theta_{\uvk}$ and $\va_{\uvk}$ are the polar and azimuthal angles of 
$\uvk$ in the laboratory frame, $D_{m'm}^l$ are Wigner's $D$-matrices, and 
\[
f_{lm}=\int_{\Sm^2}f(\uv)Y_{lm}^*(\uv)\,d\uv.
\]

\section{Case's method in three dimensions}
\label{case}

For a unit vector $\uvk$, we assume the following separated solutions with 
separation parameter $\nu$. 
\be
\psi_{\nu}(\vv{r},\uv;\uvk)=\Phi_{\nu}(\uv;\uvk)e^{-\uvk\cdot\vv{r}/\nu},
\label{case:eigen}
\ee
where the unknown function $\Phi_{\nu}(\uv;\uvk)$ is normalized as
\[
\frac{1}{2\pi}\int_{\Sm^2}\Phi_{\nu}(\uv;\uvk)\,d\uv=1.
\]
We plug the above $\psi_{\nu}(\vv{r},\uv;\uvk)$ into the homogeneous equation
\[
\left(\uv\cdot\nabla+1\right)\psi(\vv{r},\uv)
=\frac{\varpi}{4\pi}\int_{\Sm^2}\psi(\vv{r},\uv)\,d\uv,
\]
and obtain
\be
\left(1-\frac{\uvk\cdot\uv}{\nu}\right)\Phi_{\nu}(\uv;\uvk)=\frac{\varpi}{2}.
\label{case:rrf1}
\ee
Let us express $\Phi_{\nu}(\uv;\uvk)$ as
\[
\Phi_{\nu}(\uv;\uvk)=\rrf{\uvk}\phi(\nu,\mu).
\]
We will see below that $\phi(\nu,\mu)$ is independent of $\va$. Equation 
(\ref{case:rrf1}) can be rewritten as
\be
\left(1-\frac{\rrf{\uvk}\mu}{\nu}\right)\rrf{\uvk}\phi(\nu,\mu)
=\frac{\varpi}{2}.
\label{case:rrf2}
\ee
Since the right-hand side of (\ref{case:rrf2}) is a scalar, by operating 
$\irrf{\uvk}$, (\ref{case:rrf2}) reduces to
\be
\left(1-\frac{\mu}{\nu}\right)\phi(\nu,\mu)=\frac{\varpi}{2}.
\label{case:rrf3}
\ee
The above (\ref{case:rrf3}) is the equation appearing in one-dimensional 
transport theory \cite{Case60,Case-Zweifel}. Therefore it turns out that 
$\phi(\nu,\mu)$ are singular eigenfunctions, which are obtained as
\[
\phi(\nu,\mu)=\frac{\varpi\nu}{2}\mathcal{P}\frac{1}{\nu-\mu}
+\lambda(\nu)\delta(\nu-\mu),
\]
where
\[
\lambda(\nu)=1-\frac{\varpi\nu}{2}\pint_{-1}^1\frac{1}{\nu-\mu}\,d\mu
=1-\varpi\nu\tanh^{-1}(\nu).
\]
The separation constant $\nu$ takes values in $(-1,1)$ in addition to 
$\pm\nu_0$, where $\nu_0>1$ is the positive root of $\Lambda(\nu)$ 
such that $\Lambda(\nu_0)=0$. Here, the function $\Lambda(w)$ is defined 
for $w\in\Cm\setminus[-1,1]$ as
\be
\Lambda(w)=1-\frac{\varpi w}{2}\int_{-1}^1\frac{1}{w-\mu}\,d\mu.
\label{bigLambda}
\ee
Thus $\nu_0$ is given as the positive solution to the transcendental 
equation
\[
1-\varpi\nu_0\tanh^{-1}\left(\frac{1}{\nu_0}\right)=0.
\]
When $\varpi$ is near $1$, which is typical for light propagating in 
biological tissue, $\nu_0$ is approximately calculated as \cite{Case-Zweifel} 
\[
\nu_0\approx\frac{1}{\sqrt{3(1-\varpi)}}.
\]
Now we return to eigenmodes (\ref{case:eigen}) in three dimensions. We obtain
\[
\psi_{\nu}(\vv{r},\uv;\uvk)=\rrf{\uvk}\phi(\nu,\mu)e^{-\uvk\cdot\vv{r}/\nu}
=\phi(\nu,\uvk\cdot\uv)e^{-\uvk\cdot\vv{r}/\nu}.
\]
The angular flux in (\ref{rte}) is given by a superposition of eigenmodes 
$\psi_{\nu}(\vv{r},\uv;\uvk)$.

So far, the unit vector $\uvk$ is arbitrary. Hereafter we assume that $\uvk$ 
has the form
\[
\uvk=\uvk(\nu,\vv{q})=
\left(\begin{array}{c}-i\nu\vv{q}\\\hat{k}_z(\nu q)\end{array}\right),
\]
where $\vv{q}\in\Rm^2$ and $q=|\vv{q}|$. For this $\uvk$ we have
\ba
\rrf{\uvk(\nu,\vv{q})}\mu&=&
\hat{k}_z(\nu q)\mu-i\nu q\sqrt{1-\mu^2}\cos(\va-\va_{\vv{q}}),
\\
\irrf{\uvk(\nu,\vv{q})}\mu&=&
\hat{k}_z(\nu q)\mu-i|\nu q|\sqrt{1-\mu^2}\cos\va.
\ea

\begin{thm}[Orthogonality relations \cite{Machida15a}]
\label{case:orth}
For three-dimensional singular eigenfunctions we have
\[
\int_{\Sm^2}\mu\left[\rrf{\uvk(\nu,\vv{q})}\phi(\nu,\mu)\right]
\left[\rrf{\uvk(\nu',\vv{q})}\phi(\nu',\mu)\right]\,d\uv
=2\pi\hat{k}_z(\nu q)\mathcal{N}(\nu)\delta_{\nu\nu'}.
\]
Here the Kronecker delta $\delta_{\nu\nu'}$ is read as the Dirac delta 
function $\delta(\nu-\nu')$ when $\nu,\nu'$ are in the continuous spectrum 
$(-1,1)$. The normalization factor $\mathcal{N}(\nu)$ is from one-dimensional 
transport theory and given by
\[
\mathcal{N}(\nu)=\left\{\begin{aligned}
\frac{\varpi}{2}\nu^3\left(\frac{\varpi}{\nu^2-1}-\frac{1}{\nu^2}\right),
&\quad \nu=\pm\nu_0,
\\
\nu\left[\left(1-\varpi\nu\tanh^{-1}(\nu)\right)^2
+\left(\frac{\varpi \pi\nu}{2}\right)^2\right],
&\quad\nu\in(-1,1).
\end{aligned}\right.
\]
\end{thm}

\begin{proof}
For fixed $\vv{q}$, we consider two eigenvalues $\nu_1$ and $\nu_2$. 
Correspondingly, we write $\uvk_1=\uvk(\nu_1,\vv{q})$ and 
$\uvk_2=\uvk(\nu_2,\vv{q})$. We multiply (\ref{case:rrf1}) for $\nu_1$ by 
$\rrf{\uvk_2}\phi(\nu_2,\mu)$ and multiply (\ref{case:rrf1}) for $\nu_2$ by 
$\rrf{\uvk_1}\phi(\nu_1,\mu)$. We have
\ba
\left[\rrf{\uvk_2}\phi(\nu_2,\mu)\right]
\left(1-\frac{\uvk_1\cdot\uv}{\nu_1}\right)\rrf{\uvk_1}\phi(\nu_1,\mu)
&=&\frac{\varpi}{2}\rrf{\uvk_2}\phi(\nu_2,\mu),
\\
\left[\rrf{\uvk_1}\phi(\nu_1,\mu)\right]
\left(1-\frac{\uvk_2\cdot\uv}{\nu_2}\right)\rrf{\uvk_2}\phi(\nu_2,\mu)
&=&\frac{\varpi}{2}\rrf{\uvk_1}\phi(\nu_1,\mu).
\ea
By integrating both sides and subtracting the latter from former, we obtain
\[
\int_{\Sm^2}\left(\frac{\rrf{\uvk_2}\mu}{\nu_2}-\frac{\rrf{\uvk_1}\mu}{\nu_1}
\right)\left[\rrf{\uvk_1}\phi(\nu_1,\mu)\right]
\left[\rrf{\uvk_2}\phi(\nu_2,\mu)\right]\,d\uv=0.
\]
Therefore,
\[
\int_{\Sm^2}\mu\left[\rrf{\uvk_1}\phi(\nu_1,\mu)\right]
\left[\rrf{\uvk_2}\phi(\nu_2,\mu)\right]\,d\uv=0,\qquad\nu_1\neq\nu_2.
\]
When $\nu_1=\nu_2=\nu$, we have
\[
\int_{\Sm^2}\mu\left[\rrf{\uvk}\phi(\nu,\mu)\right]^2\,d\uv
=\int_{\Sm^2}\left[\irrf{\uvk}\mu\right]\phi(\nu,\mu)^2\,d\uv
=2\pi\hat{k}_z(\nu q)\int_{-1}^1\mu\phi(\nu,\mu)^2\,d\mu.
\]
The proof is completed by noticing \cite{Case-Zweifel} 
$\mathcal{N}(\nu)=\int_{-1}^1\mu\phi(\nu,\mu)^2\,d\mu$.
\end{proof}

By using three-dimensional singular eigenfunctions, let us compute the 
Green's function. When the source term in (\ref{rte}) is given as
\[
S(\vv{r},\uv)=\delta(\vv{r})\delta(\uv-\uv_0),
\]
the angular flux becomes the Green's function:
\[
G(\vv{r},\uv;\uv_0)=\psi(\vv{r},\uv).
\]
Let us calculate $G(\vv{r},\uv;\uv_0)$. Since the Green's function vanishes 
at infinity, i.e., $G(\vv{r},\uv;\uv_0)\to0$ as $|\vv{r}|\to\infty$, we 
can write $G(\vv{r},\uv;\uv_0)$ in terms of eigenmodes as
\[
G(\vv{r},\uv;\uv_0)=\left\{\begin{aligned}
\frac{1}{(2\pi)^2}\int_{\Rm^2}\left[a_+(\vv{q})\psi_{\nu_0}(\vv{r},\uv;\uvk)
+\int_0^1A_{\nu}(\vv{q})\psi_{\nu}(\vv{r},\uv;\uvk)\,d\nu\right]\,d\vv{q},
\quad z>0,
\\
\frac{-1}{(2\pi)^2}\int_{\Rm^2}\left[a_-(\vv{q})\psi_{-\nu_0}(\vv{r},\uv;\uvk)
+\int_{-1}^0A_{\nu}(\vv{q})\psi_{\nu}(\vv{r},\uv;\uvk)\,d\nu\right]\,d\vv{q},
\quad z<0,
\end{aligned}\right.
\]
where coefficients $a_{\pm}(\vv{q})$, $A_{\nu}(\vv{q})$ are determined below. 
We let $\vv{\rho}\in\Rm^2$ be the position vector in the plane perpendicular 
to the $z$-axis, i.e., $\vv{r}=(\vv{\rho},z)$ and $\vv{\rho}=(x,y)$. The jump 
condition is written as
\[
G(\vv{\rho},0^+,\uv;\uv_0)-G(\vv{\rho},0^-,\uv;\uv_0)
=\frac{1}{\mu}\delta(\vv{\rho})\delta(\uv-\uv_0).
\]
By using the above-mentioned jump condition and orthogonality relations, 
we can determine the coefficients as
\ba
a_{\pm}(\vv{q})&=&
\frac{1}{2\pi\hat{k}_z(\nu_0q)\mathcal{N}(\pm\nu_0)}
\rrf{\uvk(\pm\nu_0,\vv{q})}\phi(\pm\nu_0,\mu_0),
\\
A_{\nu}(\vv{q})&=&
\frac{1}{2\pi\hat{k}_z(\nu q)\mathcal{N}(\nu)}
\rrf{\uvk(\nu,\vv{q})}\phi(\nu,\mu_0).
\ea
Finally, the Green's function is obtained as \cite{Machida14}
\bea
G(\vv{r},\uv;\uv_0)
&=&
\frac{1}{(2\pi)^3}\int_{\Rm^2}e^{i\vv{q}\cdot\vv{\rho}}
\nonumber \\
&\times&
\Biggl\{
\frac{1}{\hat{k}_z(\nu_0q)\mathcal{N}(\nu_0)}
\left[\rrf{\uvk(\pm\nu_0,\vv{q})}\phi(\pm\nu_0,\mu)\right]
\left[\rrf{\uvk(\pm\nu_0,\vv{q})}\phi(\pm\nu_0,\mu_0)\right]
e^{\mp\hat{k}_z(\nu_0q)z/\nu_0}
\nonumber \\
&+&
\int_0^1\frac{1}{\hat{k}_z(\nu q)\mathcal{N}(\nu)}
\left[\rrf{\uvk(\pm\nu,\vv{q})}\phi(\pm\nu,\mu)\right]
\left[\rrf{\uvk(\pm\nu,\vv{q})}\phi(\pm\nu,\mu_0)\right]
e^{\mp\hat{k}_z(\nu q)z/\nu}\,d\nu
\Biggr\}\,d\vv{q},
\nonumber \\
\label{greenfunc}
\eea
where upper signs are used for $z>0$ and lower signs are chosen for $z<0$. 
We note that
\[
\rrf{\uvk(\nu,\vv{q})}\phi(\nu,\mu)
=\phi\left(\nu,\uvk(\nu,\vv{q})\cdot\uv\right).
\]

It is possible to similarly treat the case of anisotropic scattering 
\cite{Machida14}.

\section{Ganapol's Fourier-transform approach}
\label{fouriertransform}

The Green's function (\ref{greenfunc}) can also be obtained with the Fourier 
transform. We obtain \cite{Case-Zweifel,Ishimaru78}
\be
G(\vv{r},\uv;\uv_0)=
\frac{1}{r^2}e^{-r}\delta\left(\uv-\frac{\vv{r}}{r}\right)(\uv-\uv_0)
+\frac{\varpi}{2(2\pi)^4}\int_{\Rm^3}e^{i\vv{k}\cdot\vv{r}}
\frac{\left[1-\frac{\varpi}{k}\tan^{-1}(k)\right]^{-1}}
{(1+i\vv{k}\cdot\uv)(1+i\vv{k}\cdot\uv_0)}\,d\vv{k},
\label{g:conv}
\ee
where $r=|\vv{r}|$ and $k=|\vv{k}|$. The extension of (\ref{g:conv}) to 
anisotropic scattering is also possible by using rotated reference frames 
\cite{Machida15b}. In one dimension, Ganapol has developed 
an alternative Fourier-transform approach \cite{Ganapol00,Ganapol15}, which is 
different from the conventional derivation that yields the one-dimensional 
version of (\ref{g:conv}). The new formula is potentially more suitable for 
numerical calculation. Here we compute the three-dimensional Green's function 
using Ganapol's approach.

In this section we particularly set the source term as
\[
S(\vv{r},\uv)=\delta(\vv{r})\delta(\uv-\hvv{z}).
\]
The angular flux or the Green's function is then symmetric about the azimuthal 
angle. Using the Fourier transform of the Green's function
\[
\tilde{G}(\vv{k},\uv;\hvv{z})=
\int_{\Rm^3}e^{-i\vv{k}\cdot\vv{r}}G(\vv{r},\uv;\hvv{z})\,d\vv{r},
\]
we can rewrite the transport equation in the Fourier space as
\be
\left(1+i\vv{k}\cdot\uv\right)\tilde{G}(\vv{k},\uv;\hvv{z})
=\frac{\varpi}{4\pi}\int_{\Sm^2}\tilde{G}(\vv{k},\uv;\hvv{z})\,d\uv
+\delta(\uv-\hvv{z}).
\label{g:rtekspace}
\ee
Note that
\[
G(\vv{r},\uv;\hvv{z})
=\frac{1}{(2\pi)^3}
\int_{\Rm^3}\tilde{G}(\vv{k},\uv;\hvv{z})e^{i\vv{k}\cdot\vv{r}}\,d\vv{k},
\]
has the structure of separation of variables in the sense that $\tilde{G}$ 
depends on $\uv$ and $\vv{r}$ exists only in $e^{i\vv{k}\cdot\vv{r}}$. 
Let us define
\be
\tilde{\psi}_l(\vv{k})=\int_{\Sm^2}
\left[\rrf{\uvk}P_l(\mu)\right]\tilde{G}(\vv{k},\uv;\hvv{z})\,d\uv.
\label{g:moment}
\ee
We will use a new variable
\[
z=\frac{i}{k}.
\]
Noting the recurrence relation
\[
(2l+1)\mu P_l(\mu)=(l+1)P_{l+1}(\mu)+lP_{l-1}(\mu),
\]
we obtain
\be
zh_l\tilde{\psi}_l(\vv{k})-(l+1)\tilde{\psi}_{l+1}(\vv{k})
-l\tilde{\psi}_{l-1}(\vv{k})=zS_l(\uvk),
\label{psirecurr}
\ee
where
\[
h_l=2l+1-\varpi\delta_{l0},
\]
and
\[
S_l(\uvk)=(2l+1)\left.\rrf{\uvk}P_l(\mu_0)\right|_{\uv_0=\hvv{z}}.
\]
Chandrasekhar polynomials of the first and second kinds are defined as
\[
zh_lg_l(z)-(l+1)g_{l+1}(z)-lg_{l-1}(z)=0,\quad
g_0(z)=1,\quad g_1(z)=z(1-\varpi),
\]
and
\[
zh_l\rho_l(z)-(l+1)\rho_{l+1}(z)-l\rho_{l-1}(z)=0,\quad
\rho_0(z)=0,\quad \rho_1(z)=z.
\]
We can express $\tilde{\psi}_l$ as
\be
\tilde{\psi}_l=a(z)g_l(z)+b(z)\rho_l(z)+\left(1-\delta_{l0}\right)
z\sum_{j=1}^l\alpha_{l,j}(z)S_j.
\label{Eq11}
\ee
By setting $l=0$ in (\ref{Eq11}), we first notice that
\[
a(z)=\tilde{\psi}_0.
\]
By plugging (\ref{Eq11}) into (\ref{psirecurr}) we have
\[
-\left[b(z)\rho_1(z)+z\alpha_{1,1}(z)S_1\right]=z.
\]
Suppose $l>0$. Let us impose
\be
zh_l\alpha_{l,j}-(l+1)\alpha_{l+1,j}-l\alpha_{l-1,j}=0.
\label{Eq14c}
\ee
By substituting (\ref{Eq11}) for $\tilde{\psi}_l$ in (\ref{psirecurr}), we 
obtain
\[
zh_l\alpha_{l,l}S_l-(l+1)\left(\alpha_{l+1,l}S_l
+\alpha_{l+1,l+1}S_{l+1}\right)=S_l.
\]
The left-hand side of the above equation can be rewritten as
\[
\mbox{LHS}=-(l+1)\alpha_{l+1,l+1}S_{l+1}+l\alpha_{l-1,l}S_l.
\]
Hence we can put
\be
\alpha_{l-1,l}=\frac{1}{l},\qquad\alpha_{l+1,l+1}=0.
\label{Eq15}
\ee
Thus we find
\[
b(z)=-1.
\]
To find $\alpha_{l,j}(z)$, let us plug the expression 
$\alpha_{l,j}=u_jg_l+v_j\rho_l$ into $\alpha_{l,l}=0$ and 
$\alpha_{l-1,l}=1/l$. We obtain
\[
u_l=\frac{\rho_l}{l[g_{l-1}\rho_l-g_l\rho_{l-1}]},\qquad
v_l=\frac{-g_l}{l[g_{l-1}\rho_l-g_l\rho_{l-1}]}.
\]
We subtract 
$\rho_l[zh_lg_l-(l+1)g_{l+1}-lg_{l-1}]=0$ from 
$g_l[zh_l\rho_l-(l+1)\rho_{l+1}-l\rho_{l-1}]=0$, and obtain
\ba
(l+1)\left[g_l(z)\rho_{l+1}(z)-g_{l+1}(z)\rho_l(z)\right]
&=&l\left[g_{l-1}(z)\rho_l(z)-g_l(z)\rho_{l-1}(z)\right]
\\
&=&
(l-1)\left[g_{l-2}(z)\rho_{l-1}(z)-g_{l-1}(z)\rho_{l-2}(z)\right]
\\
&=&
g_0(z)\rho_1(z)-g_1(z)\rho_0(z)=z.
\ea
Thus we obtain
\[
z\alpha_{l,j}(z)=\rho_j(z)g_l(z)-g_j(z)\rho_l(z).
\]
Finally, (\ref{Eq11}) becomes
\be
\tilde{\psi}_l=g_l(z)\tilde{\psi}_0-\chi_l(\vv{k}),
\label{Eq22a}
\ee
where
\be
\chi_l(\vv{k})=\left(1-\delta_{l0}\right)\sum_{j=1}^l
\left[\rho_l(z)g_j(z)-g_l(z)\rho_j(z)\right]S_j(\uvk)+\rho_l(z).
\label{defchi}
\ee

To find the initial term $\tilde{\psi}_0$, we return to (\ref{g:rtekspace}). 
We obtain
\[
\tilde{G}(\vv{k},\uv;\hvv{z})=
\frac{\varpi}{4\pi}\frac{\tilde{\psi}_0(\vv{k})}{1+i\vv{k}\cdot\uv}
+\frac{1}{1+i\vv{k}\cdot\uv}\delta(\uv-\hvv{z}).
\]
Thus we have
\[
\tilde{\psi}_l(\vv{k})=
\frac{\varpi}{2}\tilde{\psi}_0(\vv{k})\int_{-1}^1\frac{P_l(\mu)}{1+ik\mu}\,d\mu
+\frac{1}{1+i\vv{k}\cdot\hvv{z}}
\left.\rrf{\uvk}P_l(\mu_0)\right|_{\uv_0=\hvv{z}}.
\]
Setting $l=0$, we have
\[
\left[1-\varpi L_0(z)\right]\tilde{\psi}_0(\vv{k})
=\frac{1}{1+i\vv{k}\cdot\hvv{z}},
\]
where
\[
L_l(z)=\frac{1}{2}\int_{-1}^1\frac{P_l(\mu)}{1+ik\mu}\,d\mu
=\frac{z}{2}\int_{-1}^1\frac{P_l(\mu)}{z-\mu}\,d\mu=zQ_l(z).
\]
Here, $Q_l(z)$ is the Legendre function of the second kind which has a branch 
cut from $-\infty$ to $1$. We obtain
\be
\tilde{\psi}_0=\frac{1}{\Lambda(z)}\frac{z}{z-\uvk\cdot\hvv{z}},
\label{Eq23a}
\ee
where we used
\[
1-\varpi L_0(z)=1-\varpi zQ_0(z)
=1-\frac{\varpi z}{2}\ln\frac{z+1}{z-1}
=1-\varpi z\tanh^{-1}\left(\frac{1}{z}\right)
=\Lambda(z).
\]
The function $\Lambda(z)$ is defined in (\ref{bigLambda}). 
We can calculate $\tilde{\psi}_l(\vv{k})$ using (\ref{Eq22a}) and 
(\ref{Eq23a}).

Equation (\ref{g:moment}) implies that the Fourier transform of the angular 
flux is given by
\be
\tilde{\psi}(\vv{k},\uv)
=\sum_{l=0}^{\infty}
\frac{2l+1}{4\pi}\tilde{\psi}_l(\vv{k})\rrf{\uvk}P_l(\mu).
\label{g:decompose}
\ee
When the above equation is rewritten using (\ref{Eq22a}), the dependence 
of $\tilde{\psi}_0$ becomes evident as
\[
\tilde{\psi}(\vv{k},\uv)
=\phi_{\hvv{k}}(z,\uv)\tilde{\psi}_0(\vv{k})-T_{\hvv{k}}(z,\uv),
\]
where
\ba
\phi_{\hvv{k}}(z,\uv)&=&
\sum_{l=0}^{\infty}\frac{2l+1}{4\pi}g_l(z)\rrf{\uvk}P_l(\mu),
\\
T_{\hvv{k}}(z,\uv)&=&
\sum_{l=0}^{\infty}\frac{2l+1}{4\pi}\chi_l(\vv{k})\rrf{\uvk}P_l(\mu).
\ea
Equations (\ref{Eq22a}), (\ref{Eq23a}), and (\ref{g:decompose}) yield
\ba
\psi(\vv{r},\uv)
&=&
\frac{1}{(2\pi)^3}\int_{\Rm^3}e^{i\vv{k}\cdot\vv{r}}
\tilde{\psi}(\vv{k},\uv)\,d\vv{k}
\\
&=&
\frac{1}{(2\pi)^3}\int_{\Rm^3}e^{i\vv{k}\cdot\vv{r}}\sum_{l=0}^{\infty}
\sqrt{\frac{2l+1}{4\pi}}\left[\rrf{\uvk}Y_{l0}(\uv)\right]
\left[g_l(z)\tilde{\psi}_0(\vv{k})-\chi_l(\vv{k})\right]\,d\vv{k}
\\
&=&
\frac{1}{(2\pi)^3}\sum_{l=0}^{\infty}\sum_{m=-l}^l\sqrt{\frac{2l+1}{4\pi}}
Y_{lm}(\uv)\int_{\Rm^3}e^{i\vv{k}\cdot\vv{r}}e^{-im\va_{\uvk}}
d_{m0}^l(\theta_{\uvk})\left[g_l(z)\tilde{\psi}_0(\vv{k})-\chi_l(\vv{k})\right]
\,d\vv{k}.
\ea
The above-mentioned expression can be written as
\bea
\psi(\vv{r},\uv)
&=&
\frac{1}{(2\pi)^2}\sum_{l=0}^{\infty}\sum_{m=-l}^lY_{lm}(\uv)i^m
\int_{-1}^1\int_0^{\infty}k^2
J_m\left(kr\sqrt{1-\mu_{\uvk}^2}\sin{\theta_{\hvv{r}}}\right)
e^{ikr\mu_{\uvk}\cos{\theta_{\hvv{r}}}}e^{-im\va_{\hvv{r}}}
\nonumber \\
&\times&
\kappa_{lm}(\vv{k})\,dkd\mu_{\uvk},
\label{mainresult0}
\eea
where
\[
\kappa_{lm}(\vv{k})=
\sqrt{\frac{2l+1}{4\pi}}d_{m0}^l(\theta_{\uvk})
\left[\frac{g_l(z)}{\Lambda(z)}\frac{z}{z-\mu_{\uvk}}-\chi_l(\vv{k})\right].
\]
To obtain (\ref{mainresult0}) we used
\ba
\int_0^{2\pi}e^{ikr\uvk\cdot\hvv{r}}\,d\va_{\uvk}
&=&
\int_0^{2\pi}
e^{ikr\sin\theta_{\uvk}\sin\theta_{\hvv{r}}\cos(\va_{\uvk}-\va_{\hvv{r}})}
e^{ikr\cos\theta_{\uvk}\cos\theta_{\hvv{r}}}\,d\va_{\uvk}
\\
&=&
e^{ikr\cos\theta_{\uvk}\cos\theta_{\hvv{r}}}
\int_0^{2\pi}e^{ikr\sin\theta_{\uvk}\sin\theta_{\hvv{r}}\cos\va_{\uvk}}
\,d\va_{\uvk}
\\
&=&
2\pi J_0\left(kr\sin\theta_{\uvk}\sin\theta_{\hvv{r}}\right)
e^{ikr\cos\theta_{\uvk}\cos\theta_{\hvv{r}}},
\ea
and
\[
J_0(x)=\frac{1}{2\pi}\int_0^{2\pi}e^{ix\cos\va}\,d\va.
\]

We can obtain the Green's function with this approach also 
in the case of anisotropic scattering \cite{Machida15b}.

\section{Concluding remarks}

Using the simple case of isotropic scattering in a three-dimensional infinite 
medium, we have seen how the angular flux is obtained with rotated reference 
frames. In one dimension, the solution by the singular-eigenfunction approach 
can be derived from the Fourier-transform approach \cite{Ganapol00,Ganapol15}. 
It is an interesting future problem to show the equivalence of 
(\ref{greenfunc}) and (\ref{mainresult0}).




\end{document}